\documentclass[conference]{IEEEtran}
\IEEEoverridecommandlockouts
\usepackage{cite}
\usepackage{geometry}
\usepackage{amsmath,amssymb,amsfonts}
\usepackage{hyperref}
\usepackage{subcaption}
\usepackage[lined,boxed,ruled,linesnumbered]{algorithm2e}
\usepackage{graphicx}
\usepackage{textcomp}
\usepackage{xcolor}
\usepackage{bm}
\usepackage{eucal}
\newtheorem{theorem}{Theorem}[section]
\newcommand{\C}{\mathbb{C}}
\newcommand{\R}{\mathbb{R}}
\newcommand{\diag}{\mathrm{diag}}
\newcommand{\argmax}{\mathop{\mathrm{arg\,max}}}

\usepackage{tikz}
\usetikzlibrary{arrows.meta,positioning,calc,angles,quotes}

\begin{document}
\title{\huge{Reconfigurable Intelligent Surfaces-assisted Positioning in Integrated Sensing and Communication Systems}}
\author{%
  \IEEEauthorblockN{%
    Huyen-Trang Ta\IEEEauthorrefmark{1},
    Ngoc-Son Duong\IEEEauthorrefmark{2},
    Trung-Hieu Nguyen\IEEEauthorrefmark{2},
    Van-Linh Nguyen\IEEEauthorrefmark{3}, and
    Thai-Mai Dinh\IEEEauthorrefmark{1}
  }
  \IEEEauthorblockA{\IEEEauthorrefmark{1}Faculty of Electronics and Telecommunications, VNU University of Engineering and Technology, Vietnam}
  \IEEEauthorblockA{\IEEEauthorrefmark{2}Faculty of Electronic Engineering, Posts and Telecommunications Institute of Technology, Hanoi, Vietnam}
  \IEEEauthorblockA{\IEEEauthorrefmark{3}Department of Computer Science and Information Engineering, National Chung Cheng University, Taiwan}
E-mails: $\lbrace${\ttfamily 22029064, dttmai}$\rbrace$@vnu.edu.vn, $\lbrace${\ttfamily sondn, hieunt}$\rbrace$@ptit.edu.vn, $\lbrace${\ttfamily nvlinh}$\rbrace$@cs.ccu.edu.tw
}

\maketitle
\begin{abstract}
This paper investigates the problem of high-precision target localization in integrated sensing and communication (ISAC) systems, where the target is sensed via both a direct path and a reconfigurable intelligent surface (RIS)-assisted reflection path. We first develop a sequential matched-filter estimator to acquire coarse angular parameters, followed by a range recovery process based on subcarrier phase differences. Subsequently, we formulate the target localization problem as a non-linear least squares optimization, using the coarse estimates to initialize the target's position coordinates. To solve this efficiently, we introduce a fast iterative refinement algorithm tailored for RIS-aided ISAC environments. Recognizing that the signal model involves both linear path gains and non-linear geometric dependencies, we exploit the separable least-squares structure to decouple these parameters. Furthermore, we propose a modified Levenberg algorithm with an approximation strategy, which enables low-cost parameter updates without necessitating repeated evaluations of the full non-linear model. Simulation results show that the proposed refinement method achieves accuracy comparable to conventional approaches, while significantly reducing algorithmic complexity.
\end{abstract}


\section{Introduction}
Integrated Sensing and Communication (ISAC) has emerged as a foundational element of 6G wireless networks, facilitating the seamless integration of sensing and communication functionalities to support a wide range of advanced applications, with intelligent transportation systems being a primary focus \cite{b1}. This study is dedicated to enhancing target localization accuracy through ISAC, strategically employing Reconfigurable Intelligent Surfaces (RIS) to overcome the inherent limitations and inefficiencies of conventional methods. Recently, the authors in \cite{b2} introduced a novel RIS-based passive sensing technique for angle of arrival (AOA) estimation. To address the hardware constraints of using a single-channel receiver, the authors proposed a joint framework involving atomic norm minimization (ANM) for interference mitigation and a Hankel-based multiple signal classification (MUSIC) algorithm for high-resolution AOA estimation. Their approach effectively handles communication-induced interference, which is further suppressed by optimizing the RIS measurement matrix to enhance the signal-to-interference-plus-noise ratio (SINR). Validated through theoretical Cramér-Rao lower bound (CRLB) analysis and simulations, this method demonstrates superior accuracy in vehicle localization within complex intelligent transportation scenarios. 

While passive RIS-aided ISAC systems have been widely explored, they are often limited by fading. To address this, Fang \textit{et al.} \cite{b6} investigated an active RIS-enabled ISAC framework specifically tailored for sensing extended targets in non-line-of-sight (NLoS) regions. Unlike conventional studies focusing on point-target models, the authors derived a closed-form CRLB for the target response matrix estimation. By formulating a joint optimization problem for transmit and reflective beamforming, they minimized the sensing CRB while satisfying the communication SINR requirements. The proposed alternating optimization (AO) approach, leveraging semi-definite relaxation (SDR) and successive convex approximation (SCA), demonstrates that active RIS provides significant performance gains over passive counterparts by utilizing power amplification to overcome substantial path loss. 

Recognizing that conventional MIMO systems can be cost-prohibitive, the authors in \cite{b4, b5} proposed utilizing a single-element receive antenna aided by RIS for target localization. In \cite{b4}, by manipulating the reflection phase shifts over consecutive snapshots, the RIS facilitates the generation of linearly independent signal components required for subspace-based estimation algorithms. Specifically, the authors developed a multi-object localization framework employing three RISs and introduced an interference mitigation scheme to suppress signals that do not interact with the RIS. Numerical evaluations demonstrate that the proposed method achieves high localization accuracy, with the mean square error (MSE) asymptotically approaching the CRLB as the signal-to-noise ratio (SNR) increases. Ref. \cite{b5} introduced a framework that combines an RIS with robust spatial covariance matrix (SCM) reconstruction schemes. The authors introduce a two-stage RIS phase-shifting strategy to efficiently capture spatial signal characteristics. Furthermore, the authors addressed practical hardware impairments by employing the Tikhonov regularization and total least-squares criteria for SCM reconstruction, ensuring high estimation accuracy even in the presence of phase shift errors. Compared to traditional methods, this approach significantly reduces system complexity while providing a robust solution for NLOS localization.

Unlike prior work, this work presents a novel target localization model within an RIS-assisted ISAC system that explicitly exploits dual echo paths: a direct reflection path and an RIS-assisted reflection path. The fundamental motivation for this dual-path consideration is that each echo provides complementary spatial information, which can be harnessed to significantly enhance localization precision beyond what is achievable with single-path models. To transform this multi-path information into an accurate position estimate, we adopt a hierarchical approach. The first stage begins with coarse parameter acquisition, where a sequential matched-filter and subcarrier phase analysis are employed to extract initial angular and range data. To transition from these initial estimates to a high-precision solution, we cast the localization task as a non-linear least squares (NLS) optimization. Our refinement strategy leverages the separable structure of the system model, allowing for the decoupling of linear channel gains from non-linear geometric variables. Efficiency is further ensured by a modified Levenberg update rule \cite{b9}, which utilizes an approximation strategy to bypass the exhaustive computations typically required for non-linear model updates.

\begin{figure}[t]
\centering
\begin{tikzpicture}[
    scale=0.65,
    >=Latex,
    nodeStyle/.style={circle, draw, fill=white, inner sep=1.2pt},
    lab/.style={font=\small},
    pathDir/.style={thick, dashed},
    pathRIS/.style={thick},
    measArrow/.style={-{Latex[length=2mm]}, thick}
]
\coordinate (B) at (0,0);
\coordinate (R) at (5,5);
\coordinate (T) at (10,2);
\draw[->, gray!60] (-0.5,0) -- (12,0) node[lab, below right] {$x$};
\draw[->, gray!60] (0,-0.5) -- (0,6) node[lab, above left] {$y$};
\node[nodeStyle, label={[lab]below:BS $(B)$}] (NB) at (B) {};
\node[nodeStyle, label={[lab]above:RIS $(R)$}] (NR) at (R) {};
\node[nodeStyle, label={[lab]below:Target $(T)$}] (NT) at (T) {};
\draw[pathDir] (B) -- (T);
\draw[pathDir] (B) -- (R);
\draw[pathDir] (R) -- (T);
\draw[measArrow, blue!70] ($(B)!0.08!(T)$) -- ($(T)!0.08!(B)$);
\draw[measArrow, blue!70] ($(T)!0.08!(B)$) -- ($(B)!0.08!(T)$);
\node[lab, blue!70, rotate=11.2] at ($(B)!0.55!(T) + (0.8,-0.1)$) {\tiny Direct echo $B\!\to\!T\!\to\!B$};
\draw[measArrow, red!70] ($(B)!0.08!(R)$) -- ($(R)!0.08!(B)$);
\draw[measArrow, red!70] ($(R)!0.08!(T)$) -- ($(T)!0.08!(R)$);
\draw[measArrow, red!70] ($(T)!0.08!(R)$) -- ($(R)!0.08!(T)$);
\draw[measArrow, red!70] ($(R)!0.08!(B)$) -- ($(B)!0.08!(R)$);
\node[lab, red!70, rotate=-31.5] at ($(R)!0.8!(T) + (0,0.4)$)
{\tiny RIS echo $B\!\to\!R\!\to\!T\!\to\!R\!\to\!B$};
\node[lab] at ($(B)!0.5!(T) + (0.0,0.35)$) {$d_{BT}$};
\node[lab] at ($(B)!0.5!(R) + (0,0.6)$) {$d_{BR}$};
\node[lab] at ($(R)!0.5!(T) + (0.0,-0.35)$) {$d_{RT}$};
\coordinate (Bx) at ($(B)+(2,0)$);
\pic[draw, ->, angle radius=18mm, angle eccentricity=1.2,
     "$\theta_{BT}$"] {angle = Bx--B--T};
\pic[draw, ->, angle radius=12mm, angle eccentricity=1.3,
     "$\theta_{BR}$"] {angle = Bx--B--R};
\coordinate (Rx) at ($(R)+(2,0)$);
\pic[draw, ->, angle radius=8mm, angle eccentricity=1.5, "$\theta_{RT}$"] {angle = T--R--Rx};
\draw[thick, gray!100] ($(R)+(-1.5, 0)$) -- ($(R)+(1.5, 0)$);
\end{tikzpicture}
\caption{Geometric scenario of the RIS-aided ISAC sensing with a direct path and an RIS-assisted path.}
\label{fig:1}
\end{figure}
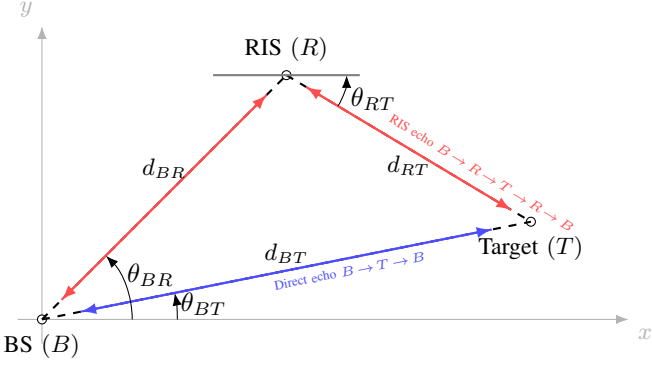

\section{System Model and problem formulation}
Consider a two-dimensional (2D) geometry with a ISAC base station (BS), an RIS, and a target located at $\mathbf{p}_B,\mathbf{p}_R,\mathbf{p}_T\in\R^2$, as shown in Fig. \ref{fig:1}. The BS employs a transmit ULA of $N_t$ elements and a receive ULA of $N_r$ elements with inter-element spacing $d=\lambda/2$, where $\lambda$ denotes the wavelength. The RIS employs a ULA of $M$ elements. Two echo sensing paths are considered: \textit{i)} Direct echo - the signal from BS touch the target and go back to BS and \textit{ii)} RIS echo - the signal from BS is transmitted to RIS then target then RIS and finally back to BS. This sensing configuration is inspired by the following theorem.
\begin{theorem}\label{th1}
If the paths are completely independent of each other, we have
\begin{equation}
    \lim_{R \rightarrow \infty} \mathrm{PEB}= 0,
\end{equation}
where $R$ and $\mathrm{PEB}$ are number of path and position error bound  (PEB), respectively.
\end{theorem}

\textit{Proof:} See Appendix

Let $d_{BT}$, $d_{BR}$, and $d_{RT}$ denote the one-way distances between the BS and target, BS and RIS, RIS and target, respectively. The two-way distances and delays are $d_{\mathrm{dir}} = 2 d_{BT}$ and $d_{\mathrm{ris}} = 2(d_{BR}+d_{RT})$, respectively. The time of arrival (TOA) would be $\tau_{\mathrm{dir}} = \frac{d_{\mathrm{dir}}}{c}$ and $\tau_{\mathrm{ris}} = \frac{d_{\mathrm{ris}}}{c}$, respectively, where $c$ is the propagation speed. The BS transmits sensing waveforms over $N$ subcarriers. We denote the subcarrier angular frequencies as $\omega_n = 2\pi (n-1)/(N T_s), n=1,\ldots,N$, where $T_s = 1/R_s$ and $R_s$ is the sampling rate. For $n$-th subcarrier, the BS transmits $N_s$ times indexed by $k = 1,\ldots,N_s$ with transmit beam vectors $\mathbf{f}_{n,k}\in\C^{N_t}$. In addition, we define the ULA steering vector for an $N_{\mathrm{ant}}$-element array as
\begin{equation}
\mathbf{a}(\theta) = \frac{1}{\sqrt{N_{\mathrm{ant}}}} \Big[1, e^{-j2\pi \frac{d}{\lambda}\sin\theta}, \ldots, e^{-j2\pi (N_{\mathrm{ant}}-1)\frac{d}{\lambda}\sin\theta}\Big]^\top.
\end{equation}
Specifically, we denote $\mathbf{a}_{t}(\theta)$ and $\mathbf{a}_{r}(\theta)$ as the BS transmit and receive steering vectors, and $\mathbf{a}_{\mathrm{RIS}}(\theta)$ as the RIS steering vector. The BS-RIS angles $\theta_{BR}$ and $\theta_{RB}$ are assumed known. The received signal at the BS for $n$-th subcarrier and $k$-th transmission time from the direct echo is modeled as
\begin{equation}
\mathbf{y}^{\mathrm{dir}}_{n,k} = \sqrt{N_tN_r}g_{\mathrm{dir}} e^{-j\omega_n \tau_{\mathrm{dir}}}\mathbf{a}_{r}(\theta_{TB})\mathbf{a}^H_t(\theta_{BT}) \mathbf{f}_{n,k},
\end{equation}
where $g_{\mathrm{dir}}$ is the complex gain and $\mathbf{f}$ is the beamforming vector. Let $\mathbf{\Phi}_{n,k} = \diag \big(e^{j\phi_{n,k,1}}, \ldots, e^{j\phi_{n,k,M}}\big)$ be the RIS's phase-shift matrix, the echo via RIS-assisted path is
\begin{equation}
\mathbf{y}^{\mathrm{ris}}_{n,k} = \sqrt{N_tN_r}g_{\mathrm{ris}} e^{-j\omega_n \tau_{\mathrm{ris}}}\mathbf{a}_{r}(\theta_{RB})\mathbf{a}^H_{t}(\theta_{BR}) \mathbf{f}_{n,k}b_{n,k}(\theta_{RT}),
\end{equation}
where $g_{\mathrm{ris}}$ is the RIS-assisted path gain and $b_{n,k}(\theta_{RT}) = \Big(\mathbf{a}_{\mathrm{RIS}}^H(\theta_{RB}) \mathbf{\Phi}_{n,k}\mathbf{a}_{\mathrm{RIS}}(\theta_{RT})\Big)^2$. Finally, the total measurement is
\begin{equation}
\mathbf{y}_{n,k} = \mathbf{y}^{\mathrm{dir}}_{n,k} + \mathbf{y}^{\mathrm{ris}}_{n,k} + \mathbf{z}_{n,k},
\end{equation}
with $\mathbf{z}_{n,k}$ being i.i.d. complex Gaussian noise. The ultimate goal is to localize the target at $\mathbf{p}_T$ by exploiting the direct and RIS-assisted echoes. Equivalently, we estimate the geometric parameters of the two sensing paths
\begin{equation}
\bm\eta \triangleq \{\theta_{BT}, \tau_{\rm d}, g_{\rm d}, \theta_{RT}, \tau_{\rm r}, g_{\rm r}\},
\end{equation}
where $\theta_{TB} = \theta_{BT}+\pi$ for the monostatic direct echo. The delays are linked to $\mathbf{p}_T$ via $\tau_{\rm{d}}=\frac{2\|\mathbf{p}_T-\mathbf{p}_B\|}{c}$ and $\tau_{\rm{r}}=\frac{2(\|\mathbf{p}_R-\mathbf{p}_B\|+\|\mathbf{p}_T-\mathbf{p}_R\|)}{c}$. The maximum-likelihood estimator is equivalent to the nonlinear least-squares (NLS) problem
\begin{equation}
\begin{aligned}
\hat{\bm{\eta}} = \arg\min_{\bm\eta} \sum_{n=1}^{N} \sum_{k=1}^{N_s} \bigg\| &\mathbf y_{n,k} - \mathbf y^{\rm dir}_{n,k}(\theta_{BT},\tau_{\rm d}, g_{\rm d}) \\
&- \mathbf y^{\rm ris}_{n,k}(\theta_{RT},\tau_{\rm r},g_{\rm r}) \bigg\|_2^2.
\end{aligned}
\label{eq:ML_NLS}
\end{equation}
Finally, $\hat{\mathbf{p}}_T$ is obtained from $\hat{\bm\eta}$ as we show latter in Sec. \ref{sec4}.

\section{Proposed Coarse Channel Estimation Method}
\subsection{Direct Path Dictionary}
For mono-static direct sensing, the AOD and AOA at the BS coincide. We thus use a single dictionary, denoted by $\CMcal{G}_\theta = \{\theta_1,\ldots,\theta_G\}, \theta_m \in [-\pi/2,\pi/2)$, to represent both the transmit and receive spatial signatures. For $n$-th subcarrier at angle $\theta_m$, we define $s^{\mathrm{dir}}_{n,k}(\theta_m) = \sqrt{N_t}\mathbf{a}^H_t(\theta_m) \mathbf{f}_{n,k}$, then, the $k$-th block of the dictionary is
\begin{equation}
\mathbf{d}^{\mathrm{dir}}_{n,k}(\theta_m) = \sqrt{N_r}\mathbf{a}_r(\theta_m+\pi) s^{\mathrm{dir}}_{n,k}(\theta_m).
\end{equation}
Stacking over $k$, we obtain the direct-path dictionary column
\begin{equation}
\bm{\omega}^{\mathrm{dir}}_{n}(\theta_m) = \big[ \mathbf{d}^{\mathrm{dir}}_{n,1}(\theta_m)^\top,\ldots,\mathbf{d}^{\mathrm{dir}}_{n,N_s}(\theta_m)^\top \big]^\top \in \C^{N_rN_s}.
\end{equation}
Collecting all atoms yields
\begin{equation}
\mathbf{\Omega}^{\mathrm{dir}}_n = [\bm{\omega}^{\mathrm{dir}}_{n}(\theta_1),\ldots,\bm{\omega}^{\mathrm{dir}}_{n}(\theta_G)] \in \C^{N_rN_s\times G}.
\end{equation}
\subsection{RIS-assisted path dictionary}
Given known $\theta_{BR}$ and $\theta_{RB}$, define $s^{\mathrm{ris}}_{n,k} = \sqrt{N_t}\mathbf{a}_t^H(\theta_{BR}) \mathbf{f}_{n,k}$ and $b_{n,k}(\theta_m) = \Big(\mathbf{a}_{\mathrm{RIS}}^H(\theta_{RB}) \mathbf{\Phi}_{n,k}\mathbf{a}_{\mathrm{RIS}}(\theta_m)\Big)^2$, the $k$-th block of the RIS dictionary is
\begin{equation}
\mathbf{d}^{\mathrm{ris}}_{n,k}(\theta_m) = \sqrt{N_r}\mathbf{a}_r(\theta_{RB}) s^{\mathrm{ris}}_{n,k}b_{n,k}(\theta_m),
\end{equation}
and the stacked column is
\begin{equation}
\bm{\omega}^{\mathrm{ris}}_{n}(\theta_m) = \big[ \mathbf{d}^{\mathrm{ris}}_{n,1}(\theta_m)^\top,\ldots, \mathbf{d}^{\mathrm{ris}}_{n,N_s}(\theta_m)^\top \big]^\top.
\end{equation}
Collecting all atoms yields
\begin{equation}
\mathbf{\Omega}^{\mathrm{ris}}_n = [\bm{\omega}^{\mathrm{ris}}_{n}(\theta_1),\ldots, \bm{\omega}^{\mathrm{ris}}_{n}(\theta_G)] \in \C^{N_rN_s\times G}.
\end{equation}
\subsection{Channel Estimation}
\subsubsection{Path gain estimation}
Let $\CMcal{P}\in\{\mathrm{d},\mathrm{r}\}$ denote the path type, with atoms $\bm{\omega}^{\CMcal{P}}_n(\theta_m)$, $m=1,\ldots,G$. We define the normalized correlation for the residual $\{\mathbf{r}_n\}_{n=1}^N$ as
\begin{equation}
\CMcal{J}^{\CMcal{P}}(m;\{\mathbf{r}_n\})
= \sum_{n=1}^N
\frac{\big|(\bm{\omega}^{\CMcal{P}}_n(\theta_m))^H \mathbf{r}_n\big|^2}
{\|\bm{\omega}^{\CMcal{P}}_n(\theta_m)\|^2}.
\end{equation}
The selected grid index and the corresponding channel coefficient are
\begin{equation}
\hat m_{\CMcal{P}} = \argmax_{m\in\{1,\ldots,G\}} \CMcal{J}^{\CMcal{P}}(m;\{\mathbf{r}_n\}),\label{angle}
\end{equation}
and
\begin{equation}
\hat h_{\CMcal{P},n}=
\frac{(\bm{\omega}^{\CMcal{P}}_n(\theta_{\hat m_{\CMcal{P}}}))^H \mathbf{r}_n}{\|\bm{\omega}^{\CMcal{P}}_n(\theta_{\hat m_{\CMcal{P}}})\|^2}, \quad \forall n, \label{gain}
\end{equation}
respectively. The reconstructed component is
\begin{equation}
\hat{\mathbf{r}}^{\CMcal{P}}_n = \bm{\omega}^{\CMcal{P}}_n(\theta_{\hat{m}_{\CMcal{P}}})\hat{h}_{\CMcal{P},n},
\end{equation}
and the residual update is $\mathbf{r}_n \leftarrow \mathbf{r}_n - \hat{\mathbf{r}}^{\CMcal{P}}_n$.
\subsubsection{Distance estimation}
The estimated channel coefficient follow $\hat{h}_{\mathrm{dir},n}  \propto g_{\mathrm{d}} e^{-j\omega_n \tau_{\mathrm{d}}}$ and $\hat{h}_{\mathrm{ris},n}  \propto g_{\mathrm{r}} e^{-j\omega_n \tau_{\mathrm{r}}}$, respectively. Herein, $\hat{h}_{\mathrm{ris},n}$ encodes the two-way distance $d_{\mathrm{r}} = c\tau_{\mathrm{r}} = 2(d_{BR}+d_{RT})$. Let $\phi_n = \angle \hat{h}_n$ for a general path, since $\phi_n \approx -\omega_n \tau + \phi_0$ and $\omega_{n+1}-\omega_{n} = 2\pi/(NT_s)$, the phase difference satisfies
\begin{equation}
\Delta \phi_n \triangleq \phi_{n+1}-\phi_n \approx -\frac{2\pi}{NT_s}\tau.
\end{equation}
Averaging over $n$ gives the delay estimator
\begin{equation}
\hat{\tau} = -\frac{NT_s}{2\pi}\frac{1}{N-1}\sum_{n=1}^{N-1} \Delta \phi_n. \label{eq19}
\end{equation}
Hence any two-way distance estimate is $\hat{d}_{\mathrm{tw}} = c\hat{\tau}$. For the direct echo, we recover the distance as
\begin{equation}
\hat{d}_{BT} = \frac{\hat{d}_{\mathrm{d,tw}}}{2}.
\label{toa1}
\end{equation}
For the RIS echo, we recover the distance as
\begin{equation}
\hat{d}_{RT} = \frac{\hat{d}_{\mathrm{r,tw}}}{2} - d_{BR},\label{toa2}
\end{equation}
\section{Proposed method for joint parameter and target position estimation}\label{sec4}
The coarse method give AOA estimates of two path, together with their TOA. To jointly mitigate grid mismatch and estimate target's position, we further refine the parameters by optimizing a least square cost function. Define the stacked vector $\mathbf{y} \in \mathbb{C}^{N_rN_sN}$ by concatenating $\mathbf{y}_{n,k}$ over $k$ then $n$ as
\begin{equation}
\mathbf{y} = [\mathbf{y}_{0,1},..., \mathbf{y}_{0,N_s},
\mathbf{y}_{1,1},..., \mathbf{y}_{1,N_s},
\mathbf{y}_{N-1,1},..., \mathbf{y}_{N-1,N_s}]^\top.
\end{equation}
It's easy to see that $(\theta_{BT},\theta_{RT},\tau_d,\tau_r)$  are not independent but are determined by the target position $\mathbf{p} = [x, y]^\top$ via geometry. We write $\theta_{BT} = \theta_{BT}(\mathbf{p})$, $\theta_{RT} = \theta_{RT}(\mathbf{p})$, $\tau_{d} = \tau_{d}(\mathbf{p})$, and $\tau_{r} = \tau_{r}(\mathbf{p})$ and re-write the sensing atom as the block corresponding to entries $(n,k)$ as
\begin{align}
\bm{\phi}_{\mathrm{d}}^{(n,k)}(\mathbf{p}) &=
\mathbf{u}_{\mathrm{rx}}\big({\theta_{BT}}(\mathbf{p})\big)e^{-j\omega_n{\tau_\mathrm{d}}(\mathbf{p})}
\big(\mathbf{u}^H_{\mathrm{tx}}({\theta_{BT}}(\mathbf{p})) \mathbf{f}_{n,k}\big),
\label{eq:phi_dir_block} \\
\bm{\phi}_{\mathrm{r}}^{(n,k)}(\mathbf{p})
&= \mathbf{u}_{\mathrm{rx}}(\theta_{\mathrm{RB}})
e^{-j\omega_n\tau_{\mathrm{r}}}(\mathbf{p})
\big(\mathbf{u}^H_{\mathrm{tx}}(\theta_{\mathrm{BR}}) \mathbf{f}_{n,k}\big)b_{n,k}\big({\theta_{RT}}(\mathbf{p})\big).
\label{eq:phi_ris_block}
\end{align}
Then, for a given position $\mathbf{p}$, define the two sensing atoms $\bm{\phi}_{\mathrm{d}}(\mathbf{p})\in\mathbb{C}^{N_rN_sN}$ and $\bm{\phi}_{\mathrm{r}}(\mathbf{p})\in\mathbb{C}^{N_rN_sN}$ such that
\begin{equation}
\mathbf{y} \approx \bm{\phi}_{\mathrm{d}}(\mathbf{p})g_{\mathrm{d}}+\bm{\phi}_{\mathrm{r}}(\mathbf{p})g_{\mathrm{r}}.
\label{eq:two_atom_model}
\end{equation}
The refinement step have to solves the NLS problem \cite{b8}
\begin{equation}
\min_{\mathbf{p}, g_{\mathrm{d}} ,g_{\mathrm{r}}}
\CMcal{J}(\mathbf{p},g_{\mathrm{d}},g_{\mathrm{r}})
\triangleq
\left\|
\mathbf{y}
-\bm{\phi}_{\mathrm{d}}(\mathbf{p})g_{\mathrm{d}}
-\bm{\phi}_{\mathrm{r}}(\mathbf{p})g_{\mathrm{r}}
\right\|_2^2.
\label{eq:nls_obj}
\end{equation}
The conventional refinement adopts a coordinate-descent (CD) update for the path gains ($g_{\mathrm{d}}, g_{\mathrm{r}}$) and a gradient-descent (GD) update for the target position ($\mathbf{p}$). This type of search can be found in \cite{b3, b7, b10}. It is simple, numerically stable, able to reach to lower bound but exhibits high complexity. To lower complexity without sacrificing the original measurement model, we propose a novel refinement method as follows.

\subsection{Finding $g_{\mathrm{d}}$, $g_{\mathrm{r}}$, given $\mathbf{p}$}
Define $\bm{\Phi}(\mathbf{p}) \triangleq \left[\bm{\phi}_d(\mathbf{p}), \bm{\phi}_r(\mathbf{p})\right] \in \mathbb{C}^{N_r N_s N \times 2}$
and $\mathbf{g}\triangleq[g_d,  g_r]^\top\in\mathbb{C}^{2\times 1}$, the closed-form solution is
\begin{equation}
\hat{\mathbf{g}} = \left(\mathbf{\Phi}(\mathbf{p})^H\mathbf{\Phi}(\mathbf{p})\right)^{-1}\mathbf{\Phi}(\mathbf{p})^H\mathbf{y}.
\label{eq:gain_ls}
\end{equation}
Compared with CD (alternating updates for $g_d$ and $g_r$), the joint LS update in \eqref{eq:gain_ls} yields faster reduction of the residual. Given $\hat{\mathbf{g}}$, the residual is
\begin{equation}
\mathbf{e} = \mathbf{y} - \mathbf{\Phi}(\mathbf{p})\hat{\mathbf{g}}.
\label{eq:residual_cost_position}
\end{equation}

\subsection{Finding $\mathbf{p}$, given $g_{\mathrm{d}}, g_{\mathrm{r}}$}
In CD manner, update of $\mathbf{p}$ would require rebuilding $\bm{\phi}_d(\theta_{BT},\tau_d)$ and $\bm{\phi}_r(\theta_{RT},\tau_r)$ at each inner iteration, which is computationally expensive due to the nested loops over subcarriers. Instead, we introduce an outer-loop base point and linearize the atoms around it. Let $\mathbf{p}_0$ be the current position estimate at the beginning of an outer iteration. Define $\theta_{BT}=\theta_{BT}(\mathbf{p}_0)$, $\theta_{RT}=\theta_{RT}(\mathbf{p}_0)$, $\tau_{d0}=\tau_d(\mathbf{p}_0)$, $\tau_{r0}=\tau_r(\mathbf{p}_0)$, we build the base atoms $\bm{\phi}_{d0} \triangleq \bm{\phi}_d(\theta_{BT},\tau_{d0})$, $\bm{\phi}_{r0} \triangleq \bm{\phi}_r(\theta_{RT},\tau_{r0})$ and their partial derivatives once as follows
\begin{subequations}
\begin{align}
\frac{\partial \bm{\phi}_d}{\partial \theta_{BT}}\Big|_0 &\triangleq \mathbf{d}_{d,\theta_{BT}}, &
\frac{\partial \bm{\phi}_d}{\partial \tau_d}\Big|_0 &\triangleq \mathbf{d}_{d,\tau},
\\
\frac{\partial \bm{\phi}_r}{\partial\theta_{RT}}\Big|_0 &\triangleq \mathbf{d}_{r,\theta_{RT}}, &
\frac{\partial \bm{\phi}_r}{\partial \tau_r}\Big|_0 &\triangleq \mathbf{d}_{r,\tau}.
\end{align}
\label{eq:base_atoms_and_partials}
\end{subequations}
These quantities correspond directly to the outputs of the functions that compute $\bm{\phi}$ and its Jacobians w.r.t. angle and delay. For any inner-loop position $\mathbf{p}$ near $\mathbf{p}_0$, define angle and delay increments $\Delta\theta_{BT} \triangleq \theta_{BT}(\mathbf{p})-\theta_{BT}$, $\Delta\theta_{RT}  \triangleq \theta_{RT}(\mathbf{p})-\theta_{RT}$, $\Delta\tau_d \triangleq \tau_d(\mathbf{p})-\tau_{d0}$, and $\Delta\tau_r \triangleq \tau_r(\mathbf{p})-\tau_{r0}$. Then the atoms are approximated by first-order Taylor expansions
\begin{equation}
\bm{\phi}_d(\theta_{BT}(\mathbf{p}),\tau_d(\mathbf{p})) \approx \bm{\phi}_{d0} + \mathbf{d}_{d,\theta_{BT}}\Delta\theta_{BT} +
\mathbf{d}_{d,\tau}\Delta\tau_d,
\label{eq:phi_dir_linear}
\end{equation}
\begin{equation}
\bm{\phi}_r(\theta_{RT}(\mathbf{p}),\tau_r(\mathbf{p}))
\approx \bm{\phi}_{r0} + \mathbf{d}_{r,\theta_{RT}}\Delta\theta_{RT} +
\mathbf{d}_{r,\tau}\Delta\tau_r.
\label{eq:phi_ris_linear}
\end{equation}
This linearization enables fast inner iterations because no subcarrier loops are required to update $\bm{\phi}_d$ and $\bm{\phi}_r$ within the inner loop. To update $\mathbf{p}$, we use an Levenberg update step that approximates the NLS objective function. The key is constructing the Jacobian of the predicted observation $\hat{\mathbf{y}}$ w.r.t. position coordinates $x$ and $y$. Let the partial derivatives be $\frac{\partial \theta_{BT}}{\partial x}$, $\frac{\partial \theta_{BT}}{\partial y}$, $\frac{\partial \tau_d}{\partial x}$, $\frac{\partial \tau_d}{\partial y}$, $\frac{\partial \theta_{RT}}{\partial x}$, $\frac{\partial \theta_{RT}}{\partial y}$, $\frac{\partial \tau_r}{\partial x}$, $\frac{\partial \tau_r}{\partial y}$, which are computed in closed-form from the 2D geometry. Within an outer iteration, we reuse the base partial derivatives from \eqref{eq:base_atoms_and_partials} and apply the chain rule
\begin{equation}
\frac{\partial \bm{\phi}_d}{\partial x}\approx \mathbf{d}_{d,\theta_{BT}}\frac{\partial \theta_{BT}}{\partial x} + \mathbf{d}_{d,\tau}\frac{\partial \tau_d}{\partial x}, \label{eq32}
\end{equation}
\begin{equation}
\frac{\partial \bm{\phi}_r}{\partial x} \approx \mathbf{d}_{r,\theta_{RT}}\frac{\partial\theta_{RT}}{\partial x} + \mathbf{d}_{r, \tau}\frac{\partial \tau_r}{\partial x}. \label{eq33}
\end{equation}
By replacing $x$ in \eqref{eq32} and \eqref{eq33} with $y$, we obtain the derivatives of the atoms with respect to $y$. Using $\hat{\mathbf{y}}=\bm{\phi}_d g_d + \bm{\phi}_r g_r$ and treating $(g_d,g_r)$ as fixed at the current inner step, we form two Jacobian columns
\begin{equation}
\mathbf{J}_x \triangleq \frac{\partial \hat{\mathbf{y}}}{\partial x} \approx \left(\frac{\partial \bm{\phi}_d}{\partial x}\right)g_d + \left(\frac{\partial \bm{\phi}_r}{\partial x}\right)g_r,
\label{eq:jacobian_column1}
\end{equation}
\begin{equation}
\mathbf{J}_y \triangleq \frac{\partial \hat{\mathbf{y}}}{\partial y}
\approx \left(\frac{\partial \bm{\phi}_d}{\partial y}\right)g_d + \left(\frac{\partial \bm{\phi}_r}{\partial y}\right)g_r.
\label{eq:jacobian_column2}
\end{equation}
Stacking them gives $\mathbf{J}\triangleq[\mathbf{J}_x\;\mathbf{J}_y]\in\mathbb{C}^{N_r N_s N\times 2}$. Because $\mathbf{p}$ is real-valued and the residual is complex, we use the standard real-part form as $\mathbf{H} \triangleq \Re\{\mathbf{J}^H\mathbf{J}\}\in\mathbb{R}^{2\times 2}$ and $\mathbf{b} \triangleq \Re\{\mathbf{J}^H\mathbf{e}\}\in\mathbb{R}^{2\times 1}$ and compute the update as
\begin{equation}
\mathbf{p}^{(i+1)} = \mathbf{p}^{(i)} + \left(\mathbf{H}+\mu\mathbf{I}_2\right)^{-1}\mathbf{b},
\label{eq:lm_step}
\end{equation}
where $\mu>0$ is the damping parameter and $\mathbf{I}_2$ is the $2\times 2$ identity matrix. 

\begin{algorithm}[t]
\caption{RIS-assisted ISAC Positioning Algorithm}
\label{alg:fast_lm}
\SetKwInOut{Input}{Input}
\SetKwInOut{Output}{Output}
\SetKwInOut{Initialize}{Initialize}

\Input{$\mathbf{y}$, $\mathbf{p}_B$, $\mathbf{p}_R$, $\theta_{BR}$,$ d_{BR}$, outer iterations $K_{out}$, inner iterations $K_{in}$}
\Output{Estimated target position $\hat{\mathbf p}$}
\BlankLine
\textit{// Coarse Estimation:}\\
Estimate angles $(\theta_{BT,0}, \theta_{RT,0})$ using $\eqref{angle}$\;
Estimate path gains $(g_{d,0}, g_{r,0})$ using $\eqref{gain}$\;
Estimate delays $(\tau_{d,0}, \tau_{r,0})$ using \eqref{eq19}\; 
\textit{// Refinement:}\\
\For{$k = 1, 2, \dots, K_{out}$}{
    Build sensing atoms and their Jacobians using \eqref{eq:base_atoms_and_partials}\;

    Update sensing atoms using \eqref{eq:phi_dir_linear} and \eqref{eq:phi_ris_linear}\;
    
    \For{$i = 1, 2, \dots, K_{in}$}{
        Update path gains using \eqref{eq:gain_ls}\;
        
        Compute residual using \eqref{eq:residual_cost_position}\;
        
        Construct linearized Jacobian using \eqref{eq:jacobian_column1}-\eqref{eq:jacobian_column2}\;
        
        Update $\mathbf{p}$ using \eqref{eq:lm_step}\;
    }
}
\Return $\hat{\mathbf{p}}$\;
\end{algorithm}
\section{Simulation Results and Discussion}
\subsection{Simulation Setup}
We consider a 2D monostatic RIS-aided ISAC sensing scenario as shown in Fig. \ref{fig:1}. $\theta_{BR}$ and $d_{BR}$ are assumed known. The node coordinates are fixed as $\mathbf{p}_B = [0, 0]^\top, \mathbf{p}_R = [5, 5]^\top$ and $\mathbf{p}_T$ is chosen randomly in the first quadrant. The true angle and distance then can be inferred by the geometrical relationship between nodes. The BS employs co-located transmit/receive ULAs with $N_t = N_r = 32$ antenna elements, and the RIS is a ULA with $M = 32$ elements. The inter-element spacing is $d = \lambda/2$. The propagation speed is set to $c = 300 ~\mathrm{m}/\mu\mathrm{s}$. The BS transmits sensing waveforms over $N = 20$ subcarriers with sampling rate $R_s = 100 ~ \mathrm{MHz}$ and sampling period $T_s = 1/R_s$. For each subcarrier $n$, the BS probes the scene over $N_s = 32$ transmitted time. For each realization, we generate random unit-modulus precoders $\bm{f}_{n,k}\in\mathbb{C}^{N_t}$ with $[\bm{f}_{n,k}]_i=e^{j2\pi u}$ where $u \sim \CMcal{U}(0,1)$ independently for all $i,n,k$. We also create random RIS phase matrices $\Phi_{n,k} = \mathrm{diag}(e^{j\varphi_{n,k,1}},\ldots,e^{j\varphi_{n,k,M}})$ with $\varphi_{n,k,m}\sim\CMcal{U}(0,2\pi)$ i.i.d. for all $m,n,k$. All results are averaged over $1000$ independent realizations. We use \textit{i)} position RMSE and \textit{ii)} algorithmic complexity as two performance metrics.
\subsection{Simulation results and discussion}
\subsubsection{On the coarse estimation}
\begin{figure}[t]
    \centering
    \begin{subfigure}{0.9\columnwidth} 
        \centering
        \includegraphics[width=\linewidth]{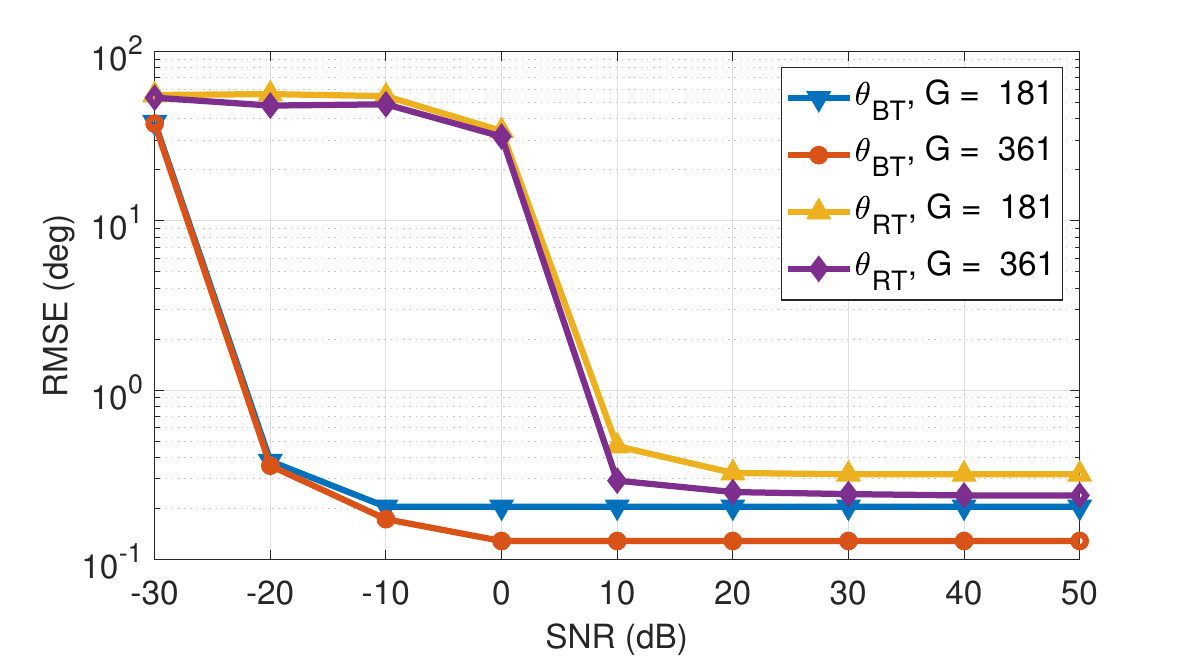}
        \caption{Angle's RMSE vs SNR.}
        \label{fig:2a}
    \end{subfigure}

    \begin{subfigure}{0.9\columnwidth}
        \centering
        \includegraphics[width=\linewidth]{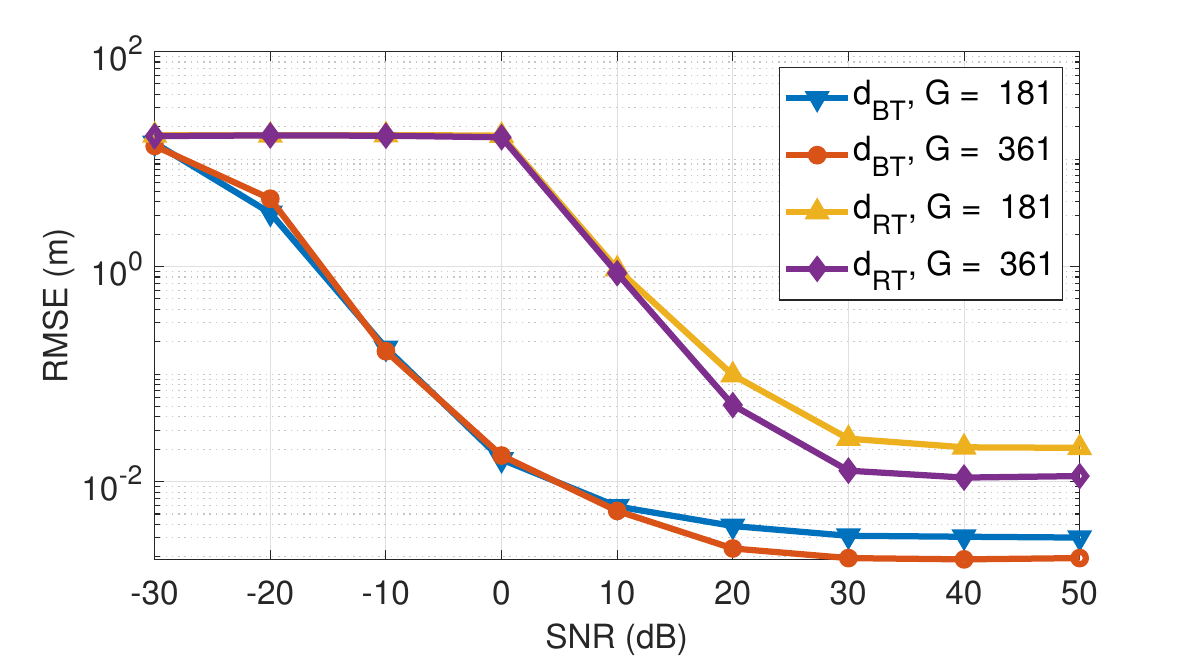}
        \caption{Range's  RMSE vs SNR.}
        \label{fig:2b}
    \end{subfigure}
    
    \caption{Coarse estimate evaluation}
    \label{fig:2}
\end{figure}
Fig. \ref{fig:2} summarizes the coarse estimation performance of the proposed two–stage scheme as a function of SNR and dictionary size $G$. As anticipated, the direct path $(\theta_{BT}, d_{BT})$ is always estimated more accurately than the RIS path $(\theta_{RT}, d_{RT})$. This gap is due to the RIS estimator operates under an effectively lower SNR and in the presence of colored noise from the imperfect cancellation of the direct path. The impact of the angular grid size is also clearly visible. For both paths, increasing $G$ from $181$ to $361$ yields a reduction of the high-SNR RMSE, in particular for the angles. This behavior is consistent with the fact that the quantization floor scales approximately as $\Delta\theta$, where $\Delta\theta = \pi/(G-1)$ is the angular grid spacing. A finer grid reduces the deterministic approximation error of the on-grid model and simultaneously mitigates the residual interference leaked from the first stage. Consequently, the effective noise seen by the RIS path is also reduced when $G$ is larger. 

Another important scene is that, for the same $G$, the high-SNR RMSE of the RIS path does not converge to that of the direct path. Instead, the RIS curves saturate at a higher level and at a larger SNR. This delayed and elevated saturation is compatible with the sequential error-propagation, even when the thermal noise becomes negligible, the RIS estimator still suffers from the structured residual created by the coarse direct-path estimation. In range (Fig. \ref{fig:2b}), this effect is amplified because the RIS delay is inferred from the phase slope of the composite cascaded coefficient, so any small angular mismatch in the first stage translates into a non-vanishing perturbation of the phase across subcarriers.

\subsubsection{On the refinement estimation}
Fig. \ref{fig:3} shows the effectiveness of the proposed refinement method in terms of convergence speed and localization accuracy. As shown in Fig. \ref{fig:3a}, both methods ultimately converge to the correct magnitudes. The proposed method shows mild non-monotonic fluctuations during early iterations. This behavior is consistent with a more aggressive update mechanism of Levenberg algorithm and the use of local linearization, which may yield larger corrective steps but is subsequently stabilized as the iterates enter the local basin of attraction. In contrast, CD+GD tends to exhibit a smoother toward the exact value. This behavior is also observed in the RIS-assisted path.
\begin{figure}[t]
    \centering
    \begin{subfigure}{0.93\columnwidth} 
        \centering
        \includegraphics[width=\linewidth]{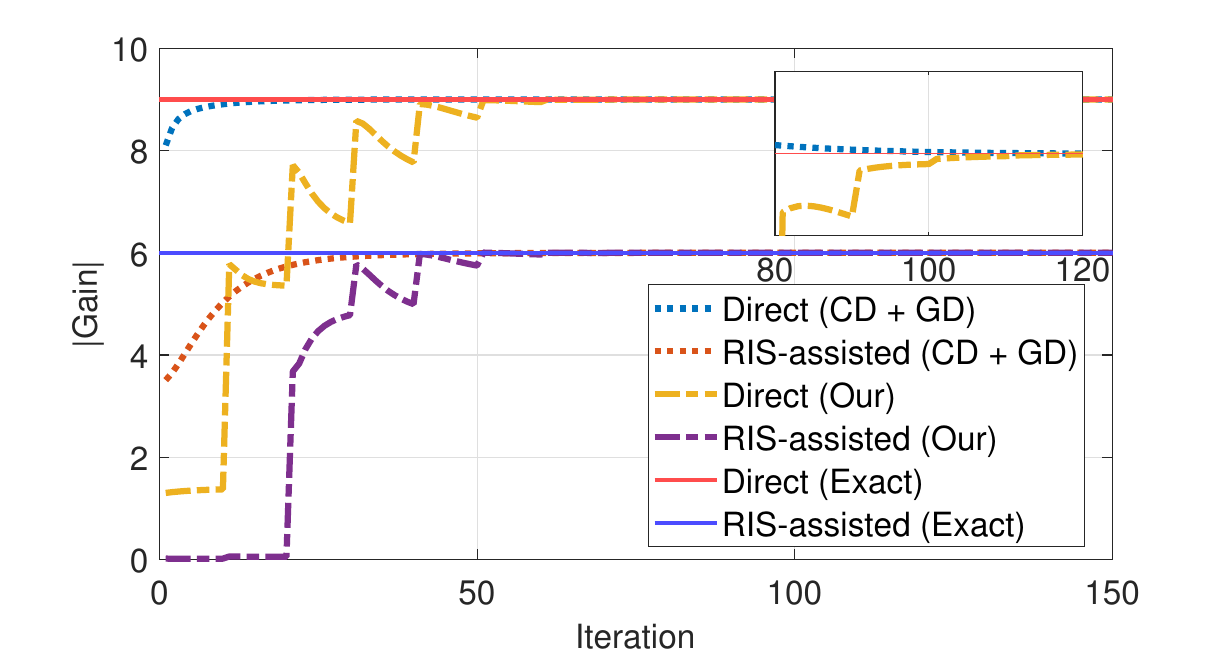}
        \caption{Convergence evaluation in the noiseless case.}
        \label{fig:3a}
    \end{subfigure}
    \begin{subfigure}{0.93\columnwidth}
        \centering
        \includegraphics[width=\linewidth]{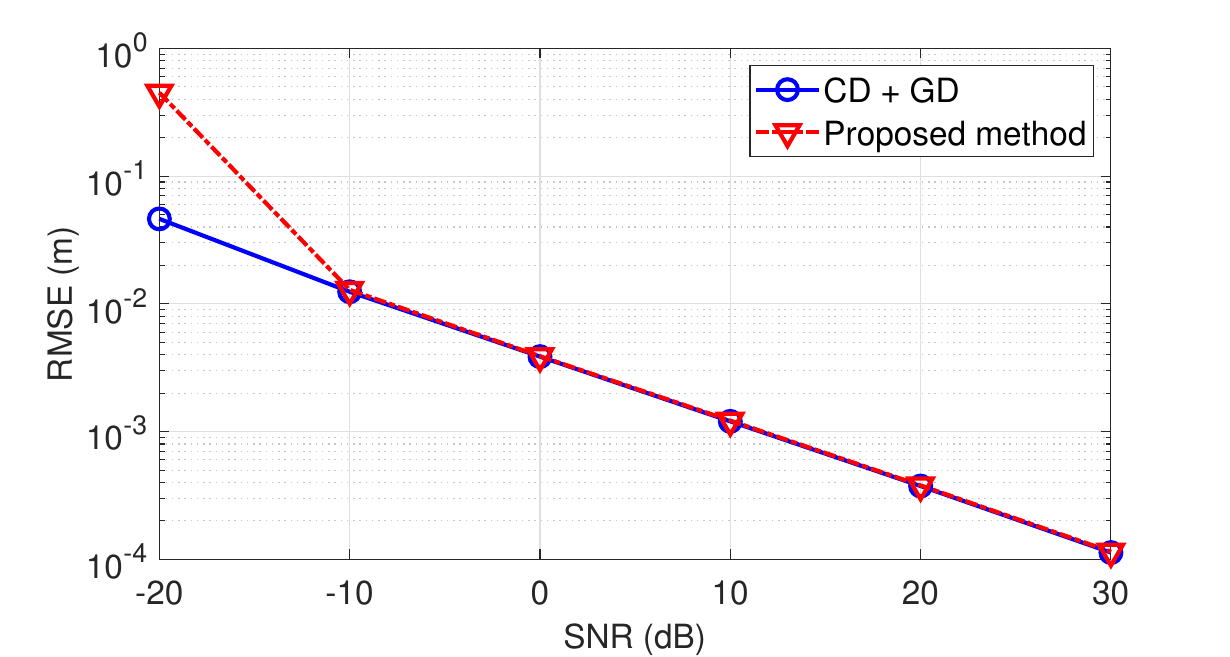}
        \caption{Position's RMSE vs SNR.}
        \label{fig:3b}
    \end{subfigure}
    
    \caption{Refinement evaluation}
    \label{fig:3}
\end{figure}
Fig. \ref{fig:3b} shows RMSE versus SNR on a log-scale for both CD+GD and the proposed refinement method. It is clear that for moderate-to-high SNR (approximately -10 dB and above), the two curves are almost indistinguishable and decay with essentially the same slope, indicating that the proposed method achieves the same noise-limited localization accuracy as the conventional method. In particular, the performance of the proposed method at high SNR suggest that the linearized Levenberg-based refinement does not introduce a persistent model-mismatch error; instead, it converges to an estimator with the same effective asymptotic behavior as CD+GD. 

At very low SNR (notably at -20 dB), the proposed method exhibits a markedly larger RMSE than CD+GD, which is consistent with the expected sensitivity of local linearization when the iterates are more strongly perturbed by noise. As the SNR increases, this gap rapidly vanishes and both methods enter the same regime where estimation performance is dominated by measurement noise rather than algorithmic effects, leading to practically identical RMSE performance. Also, it is anticipated that the performance of both algorithm go linearly (by log-scale) as SNR increase.

\subsubsection{On the complexity}
Let $L = N_rN_sN$ denote the length of the stacked measurement vector. In the CD+GD refinement, each iteration requires \textit{i)} rebuilding the two sensing atoms and their Jacobians, whose dominant cost scales as $C_{\mathrm{build}} = \CMcal{O}\big(NN_s(2N_t + M + 2N_r)\big)$ due to the per-$(n,k)$ beamforming inner-products and RIS combining, and \textit{ii)} lightweight updates (LS for gains) whose cost is $\CMcal{O}(L)$. Hence, over $K$ iterations, the overall complexity is $\CMcal{O}\big(KC_{\mathrm{build}}\big)$ with  $C_{\mathrm{build}}\gg L$. In contrast, the proposed method amortizes the expensive rebuild by performing $K_{in}$ inner updates per outer rebuild: each outer iteration incurs one rebuild of cost $C_{\mathrm{build}}$, while each inner step relies only on linearized atom updates, a $2\times 2$ joint LS, and a $2\times 2$ Levenberg solve, which together scale as $\CMcal{O}(L)$. Therefore, for $S = K_{out}K_{in}$ inner steps, the total complexity is $\CMcal{O}\big(K_{out}C_{\mathrm{build}} + S L\big)$. When the atom/Jacobian construction is the computational bottleneck, this yields an effective speedup on the order of $K_{in}$.
\section{Conclusion}
This paper introduces a framework for high-precision target sensing in RIS-assisted ISAC systems. We proposed a robust two-stage estimation strategy comprising a coarse initialization via sequential matched-filtering and a fine-tuning stage using a fast iterative refinement algorithm. By exploiting the separable least-squares structure and a modified Levenberg algorithm, we effectively decoupled the linear path gains from the non-linear geometric dependencies. Simulation results demonstrate that the proposed method achieves accuracy comparable to conventional non-linear estimators but with significantly reduced computational complexity. Future work will extend this framework to broader scenarios, specifically addressing the challenges of multi-target interference and tracking moving targets in dynamic environments.
\section*{Appendix I: Proof of Theorem \ref{th1}}
Defining $\hat{\bm{\eta}}$ as the unbiased estimator of $\bm{\eta}$, its MSE is bounded as
\begin{equation}
\mathbb{E}_{\bm{y}|\bm{\eta}} [(\hat{\bm{\eta}} - \bm{\eta})(\hat{\bm{\eta}} - \bm{\eta})^\top] \succeq \bm{J}^{-1}_{\bm{\eta}},
\end{equation}
where $\mathbb{E}_{\bm{y}|\bm{\eta}}[.]$ denotes the expectation parameterized by the unknown parameters $\bm{\eta}$, and $\bm{J}_{\bm{\eta}}$ is the FIM, defined as
\begin{equation}
\mathbf{J}_{\bm{\eta}} \triangleq \mathbb{E}_{\bm{y}|\bm{\eta}} \left[-\frac{\partial^2 \ln f(\bm{y}|\bm{\eta})}{\partial\bm{\eta} \partial\bm{\eta}^\top}\right],
\end{equation}
where $f(\bm{y}|\bm{\eta})$ is the likelihood function of $\bm{y}$ conditioned on $\bm{\eta}$. We determine the position's Fisher Information Matrix (FIM) through a transformation of variables from $\bm{\eta}$ to $\tilde{\bm{\eta}} = [x, y]^\top$, where $x$ and $y$ are the target's position coordinates in 2D. The FIM of $\tilde{\bm{\eta}}$ is obtained by means of the transformation matrix $\mathbf{T}$ as
\begin{equation}
\mathbf{J}_{\tilde{\bm{\eta}}} =  \mathbf{T}^\top\mathbf{J}_{\bm{\eta}}\mathbf{T},
\end{equation}
where $\mathbf{T} \triangleq \partial\bm{\eta}^\top/\partial\tilde{\bm{\eta}}$. Since the paths are assumed to be independent, the FIM for the channel parameters, $\mathbf{J}_{\bm{\eta}}$, can be expressed as the sum of the FIMs corresponding to each path
\begin{equation}
\mathbf{J}_{\bm{\eta}} = \mathbf{J}_{\bm{\eta},1} + \mathbf{J}_{\bm{\eta},2} + \dots + \mathbf{J}_{\bm{\eta},R},
\end{equation}
where $\mathbf{J}_{\bm{\eta},r} \succeq 0$ represents the FIM contribution from all path. This decomposition holds because the independence of the paths implies that the log-likelihood function is additive across paths, leading to additive Fisher information. Consequently, the summand, $\mathbf{J}_{\bm{\eta}}$, is also positive semi-definite. Applying the position transformation, the FIM for $\tilde{\bm{\eta}} = [x, y]^\top$ becomes:
\begin{equation}
\begin{split}
\mathbf{J}_{\tilde{\bm{\eta}}} = \mathbf{T}^\top \mathbf{J}_{\bm{\eta}} \mathbf{T} 
&= \mathbf{T}^\top (\mathbf{J}_{\bm{\eta},1} + \mathbf{J}_{\bm{\eta},2} + \dots + \mathbf{J}_{\bm{\eta},R}) \mathbf{T} \\
&= \mathbf{J}_{\tilde{\bm{\eta}},1} + \mathbf{J}_{\tilde{\bm{\eta}},2} + \dots + \mathbf{J}_{\tilde{\bm{\eta}},R} \succeq 0,
\end{split}
\end{equation}
where $\mathbf{J}_{\tilde{\bm{\eta}},r} = \mathbf{T}^\top \mathbf{J}_{\bm{\eta},r} \mathbf{T}$ is the FIM contribution in the position space for the $r$-th path. Since each $\mathbf{J}_{\bm{\eta},r}$ is positive semi-definite, and $\mathbf{T}$ is the Jacobian matrix of the transformation, each $\mathbf{J}_{\tilde{\bm{\eta}},r}$ is also positive semi-definite. Thus, $\mathbf{J}_{\tilde{\bm{\eta}}}$ is positive semi-definite as the sum of positive semi-definite matrices. Hence, as the CRB is defined as the inverse of the FI, then when the number of path increases, the PEB approaches zero.

\end{document}